\documentclass[journal]{IEEEtran}

\usepackage{cite}
\ifCLASSINFOpdf
   \usepackage[pdftex]{graphicx}
\else
\fi

\usepackage{amsmath,amssymb,amsfonts,amsthm,mathrsfs}
\usepackage{xcolor}
\usepackage{array}

\usepackage{url}
\newtheorem{definition}{Definition}

\newtheorem{theorem}{Theorem}
\newtheorem{lemma}{Lemma}

\newtheorem{example}{Example}
\newtheorem{proposition}{Proposition}
\newtheorem{remark}{Remark}

\begin{document}

\title{On a Convex Logic Fragment \\ for Learning and Reasoning}

\author{Francesco~Giannini, Michelangelo~Diligenti, Marco~Gori, and~Marco~Maggini

\thanks{Manuscript received September 15, 2018; revised October 31, 2018.}
\thanks{The authors are with the University of Siena, via Roma 56, 53100 ITALY, (e-mail: \{fgiannini,diligmic,marco,maggini\}@diism.unisi.it)}%

}

\markboth{Paper ID TFS-2018-0688}%
{Shell \MakeLowercase{\textit{et al.}}: Bare Demo of IEEEtran.cls for IEEE Journals}



\IEEEoverridecommandlockouts
\IEEEpubid{\makebox[\columnwidth]{\copyright Copyright 2018 IEEE \hfill} \hspace{\columnsep}\makebox[\columnwidth]{ }}

\maketitle

\begin{abstract}
In this paper we introduce the \emph{convex fragment} of \L ukasiewicz Logic and discuss its possible applications in different learning schemes. Indeed, the provided theoretical results are highly general, because they can be exploited in any learning framework involving logical constraints. The method is of particular interest since the fragment guarantees to deal with convex constraints, which are shown to be  equivalent to a set of linear constraints. Within this framework, we are able to formulate learning with kernel machines as well as collective classification as a quadratic programming problem.
\end{abstract}


\IEEEpeerreviewmaketitle

 \IEEEpubidadjcol

\section{Introduction}
\label{sec:introduction}

\IEEEPARstart{T}{he} theory we present can be exploited in different learning settings, especially in contexts where some relational knowledge on the task is available. In general, a learning process can be thought of as a constraint satisfaction problem, where the constraints represent the knowledge about the functions to be learned. 
Supervisions act as a special class of constraints providing positive as well as negative examples for the task, while it may be also useful to express relationships among the classification categories. For example, we can be interested in the satisfaction of a rule like \emph{``any pattern classified as a cat has to be classified as an animal''}, where \emph{cat} and \emph{animal} have to be thought of as the membership functions of two classes to learn. In such a sense, symbolic logic provides a natural way to express factual and abstract knowledge about a problem by means of logical formulas in a certain logic. Logical representations have been successfully employed in various learning schemes from a long time, since they allow high--level representations. In addition, we can exploit theorems and fundamental properties of the chosen logic to get an advantage for the learning strategy. In fact, we choose \L ukasiewicz Logic since the McNaughton Theorem provides a functional representation of \L ukasiewicz formulas by piecewise linear functions. Exploiting this result, we are able to characterize the fragment of \L ukasiewicz formulas corresponding to convex functional constraints (see also \cite{giannini2017learning}). As we better clarify in the following, this result is very general because it can be applied to different learning settings, where the employment of the fragment brings benefits to the problem solution. In particular, this paper shows how to exploit this result into kernel machines, collective classification and Probabilistic Soft Logic. In the experimental section, the theoretical result is applied to a simple artificial learning task to enlighten the properties of the proposed framework in a controlled setting.

The paper is organized as follows: Section~\ref{sec:learning_and_convexity} discusses the importance of expressing convex constraints representing logic knowledge for machine learning applications, while in Section \ref{sec:relwork} we discuss related works and we carry out an informal comparison with respect to the main approaches making use of logical constraints in the literature. Section \ref{sec:luk} reports the theoretical results characterizing the \L ukasiewicz fragment that yields convex logical constraints. In Section \ref{sec:zoo} we introduce different learning schemes where the fragment can be successfully applied. For instance, we show how to extend classical kernel machine theory with logical constraints still preserving quadratic optimization. Further, we formulate a collective classification task and we discuss the extension of Probabilistic Soft Logic allowing logical formulas to belong to the convex fragment instead of be restricted to disjunctive clauses. In section \ref{sec:exp} the theory is evaluated on an experimental setting. Finally, some conclusions and additional remarks are drawn in Section \ref{sec:conc}. The Appendix \ref{appendix} reports the proofs of the theoretical results, together with more technical details.

\section{Learning and Convexity}
\label{sec:learning_and_convexity}
The field of Machine Learning (ML)~\cite{robert2014machine,nasrabadi2007pattern} defines theories and a wide range of techniques, which can be used to approximate an unknown function. A classical starting point for most ML approaches is to define a cost functional (the same reasoning can be applied to probabilistic methods, where the cost functional is replaced by a likelihood). The cost functional depends on the parameters of the approximators and it assumes lower values for functions having a closer-to-desired behavior. For example, the cost functional can express how the function should behave on some data points or it can penalize values outside a given range or force the function to respond similarly to different pairs of inputs.
Once the cost functional is defined, the training of the learner becomes an optimization problem with respect to the parameters, which can be tackled via gradient descent or other optimization techniques.
One desired property of a cost functional is to be convex. Indeed, non-convex optimization is intractable in a general sense~\cite{szu1986non} and it can be efficiently faced only when sub-optimal solutions are acceptable\footnote{An optimal solution is a solution for which there are no other solutions providing a lower cost.}.
On the other hand, convex optimization is very well understood and there are available algorithmic solutions to efficiently find optimal solutions. For example, Conjugate Gradient Descent~\cite{moller1993scaled} is guaranteed to converge to an optimal solution in linear time when the cost functional is convex.
Many methods in machine Learning explicitly aim at defining convex cost functionals by constraining the form of the approximation. For example, training of Support Vector Machines~\cite{cortes1995support} corresponds to solving a linear (and therefore convex) optimization problem. However, it is not obvious how to get a convex cost functional in general. For instance, the definition of a convex cost functional is more challenging when the learning task is more complex and general logic knowledge needs to be exploited to  express properties of the unknown function, like assumed by Statistical Relational Learning (SRL) methods~\cite{getoor2007introduction}.
This class of learners includes methods like Markov Logic Networks~(MLNs)~\cite{richardson2006markov} defining a probabilistic logic, which assigns a weight to each clause in the knowledge base. However, inference in MLNs builds a non-convex optimization problem, making the methodology impractical for large learning problems.
To overtake this limitation, Probabilistic Soft Logic (PSL)~\cite{kimmig2012short} relaxes the learning task using fuzzy logic and restricts the knowledge base to clauses with a specific form. Under these constraints inference corresponds to a convex optimization problem, which can be solved even for large learning tasks. However, the limitations in the form of the clauses prevents the application of PSL to arbitrary learning tasks.
This paper studies under which conditions a general description of the knowledge base corresponds to a convex optimization problem, when integrated into learning. This class of descriptions is larger than that considered in PSL and similar learning methods, making this theoretical result useful across a wide range of machine learning applications.

\section{Related Works}
\label{sec:relwork}
There are several approaches to embed logical knowledge into learning processes.
As an example, the connections between logic and \emph{kernel machines} have been the subject of many investigations with different techniques from several authors. For instance, in \cite{cumby2003kernel} a family of kernel functions is built up from a Feature Description Language to exploit the relational structure of the domain. One limitation of this work is that the integration with the logic formalism does not reveal very tight connections. A different approach is considered in \cite{diligenti2012bridging}, where the t-norm theory is used to translate first--order formulas into real valued functions. Unfortunately, the logical constraints turn out to be not convex in general, unless one restricts the attention to only Horn clauses~\cite{diligenti2016learning}. With respect to this point, the fragment we present in this paper provides the complete set of formulas that can be written in \L ukasiewicz logic to get convex functional constraints. 

A different research area that combines logic programming with machine learning techniques is \emph{Inductive Logic Programming} (ILP) \cite{muggleton1991inductive,muggleton1994inductive}. The general inductive problem is as follows: given a set of positive $P$ and negative $N$ examples and a consistent background knowledge $B$, find a hypothesis $H$ such that the conjunction of $H$ and $B$ entails all the examples of $P$ and none of $N$. A large number of hypotheses typically fits such a definition. For instance the Bayesian ILP setting \cite{muggleton1994bayesian} assumes a prior probability distribution defined over the hypothesis space. In \cite{de2008probabilistic} clauses are given a probability value and two methods to estimate these parameters and the hypothesis are provided.  In addition, it is worth to mention some related works on Inductive Logic Programming and kernel machines like \cite{landwehr2010fast} and \cite{muggleton2005support}. In the first paper the system FOIL is combined with kernel methods by leveraging FOIL search for a set of relevant clauses. In the second one a kernel, that is an inner product in the feature space spanned by a given set of first--order hypothesized clauses, is proposed.

\emph{Statistical Relational Learning} (SRL) deserves a special mention in this brief overview. SRL focus on relational domains under uncertainty, where relations among objects are expressed by first--order formulas and uncertainty is handled by setting up probabilistic graphical models, like Bayesian networks. Probabilistic logic learning has been the subject of many investigations from several authors and we recommend \cite{de2003probabilistic} for a survey. In these years, a lot of frameworks arose into this field, among which emerged \emph{Markov Logic Networks} (MLNs) \cite{richardson2006markov} and \emph{Probabilistic Soft Logic} (PSL) \cite{kimmig2012short} already mentioned. Formally an MLN is a set of weighted First--Order logic (FOL) formulas, but can be viewed as a template for constructing Markov Random Fields (MRFs) to model the joint distribution of the set of all the possible atomic groundings in formulas. PSL, initially called Probabilistic Similarity Logic \cite{brocheler2012probabilistic}, has a very similar approach, that can be viewed as a template for Hinge-Loss Markov Random Fields, that are continuous generalization of MRFs whose formulas are restricted to disjunctive clauses and translated by the \L ukasiewicz t-norm and t-conorm. With this restriction, the authors are able to get convex optimization \cite{bach2013hinge} for the \emph{most probable explanation} (MPE) task avoiding the general intractability of MLNs and they also provide different approaches to estimate the rule weights \cite{bach2015hinge}. As we point out in section \ref{sec:zoo}, where a more accurate comparison with PSL is discussed, the set of formulas, that keep convexity in this frame, can be extended to the whole convex fragment we provide. However, our approach is more in the spirit of \cite{diligenti2017semantic} where Semantic Based Regularization (SBR), a unified framework for inference and learning that is centered around the notion of constraints and of parsimony principle,  is presented. The SBR goal is to find the smoothest functions satisfying the (possibly weighted) constraints. As pointed out by the paper, the given solution can be interpreted also in probabilistic terms and directly compared with MLNs. Finally, in a similar direction, some recent works by Serafini et al. \cite{serafini2016logic,serafini2016learning} propose a framework called Logic Tensor Networks (LTN) that integrates learning based on tensor networks \cite{socher2013reasoning} with reasoning using first-order \L ukasiewicz Logic, all implemented in TENSORFLOW.

\section{The Convex \L ukasiewicz Fragment}
\label{sec:luk}
In this section we present the main theoretical contribution of the paper. Since the result is about logic, it applies to a wide class of learning settings exploiting logical arguments. In particular, we discuss the general case of a multi--task learning problem, where we are interested in the integration of prior knowledge with supervisions for learning a set $\{p_1,\ldots,p_J\}$ of real--valued functions. We decided to represent the prior knowledge by means of FOL formulas giving that this logical formalism allows us to express elaborated relations among the objective functions in a simple way and with a suitable granularity. Since the objectives can assume continuous values, a fuzzy logic turns out to be an effective tool to express the logical constraints and like other authors do \cite{serafini2016logic,bach2015hinge}, we decided to focus on \L ukasiewicz Logic. In principle, the logical constraints can be expressed by any (fuzzy) logic, however there are several reasons to work into the \L ukasiewicz frame. For instance, among the three fundamental fuzzy logics given by a continuous t-norm\footnote{A t-norm  is a binary operation extending the boolean conjunction to [0,1].} (\L ukasiewicz, G\"odel and Product), the \L ukasiewicz logic is the only one providing an equivalent \emph{prenex normal form} (i.e. quantifiers followed by a quantifier--free part) and a continuous involutive negation ($\neg\neg x = x$) preserving the De Morgan laws (for more details on involutive negations in fuzzy logic, see \cite{esteva2000residuated}). 

\begin{remark}
	For the learning problem we can assume to deal with propositional formulas. Indeed, the objective functions are evaluated on finite training sets and according to the following rules, FOL formulas can be rewritten in an equivalent free--quantifier form as:
	\begin{equation}\label{qq}
	\forall x\, p_i(x)\simeq\bigwedge_{a\in Dom(p_i)} p_i(a), \quad \exists x\, p_i(x)\simeq\bigvee_{a\in Dom(p_i)} p_i(a).
	\end{equation}
	
	As a result, we can exploit stronger results from propositional theories still yielding FOL expressiveness.
\end{remark}

\subsection{Propositional \L ukasiewicz Logic}
Propositional \L ukasiewicz Logic {\bf \L} is the fuzzy logic we get if we take the t--norm $x\otimes y=\max\{0,x+y-1\}$ as truth function for the conjunction on the continuous values $[0,1]$. It is worth to mention that {\bf \L} is sound and complete with respect to its standard algebra on $[0,1]$ and the operations corresponding to \L ukasiewicz connectives are reported in the Table \ref{tab:luksem}. We only notice that the algebraic semantics of \L ukasiewicz connectives is determined by the chosen t--norm \cite{hajek1998metamathematics} and we recommend \cite{novakmathematical} for a more detailed analysis about \L ukasiewicz logic.

\begin{table}[htbp]
		\caption{\L ukasiewicz connectives and their algebraic counterparts.}
	\centering
		\begin{tabular}{|c|c|}
			\hline
			{\bf Formula} & {\bf Operation} \\
			\hline
			\hline
			\multicolumn{2}{|c|}{\emph{conjunctions}} \\
			\hline
			$x\otimes y$ & $\max\{0,x+y-1\}$ \\
			\hline
			$x\wedge y$ & $\min\{x,y\}$ \\
			\hline
			\hline
			\multicolumn{2}{|c|}{\emph{disjunctions}} \\
			\hline
			$x\oplus y$ & $\min\{1,x+y\}$ \\
			\hline
			$x\vee y$ & $\max\{x,y\}$ \\
			\hline
			\hline
			\multicolumn{2}{|c|}{\emph{implication and negation}} \\
			\hline
			$x\rightarrow y$ & $\min\{1,1-x+y\}$ \\
			\hline
			$\neg x$ & $1-x$ \\
			\hline
			\hline
			\multicolumn{2}{|c|}{\emph{constants}} \\
			\hline
			$\underbar{0}$ & $0$ \\
			\hline
			$\underbar{1}$ & $1$ \\
			\hline
		\end{tabular}
	\label{tab:luksem}
\end{table}
The connectives $\otimes,\oplus$ correspond to \L ukasiewicz t--norm and t--conorm. They are called \emph{strong} connectives in opposition to $\wedge,\vee$ that are called \emph{weak}. Indeed the following relations hold for all $x,y\in[0,1]$:
\[
x\otimes y\leq x\wedge y\leq x\vee y\leq x\oplus y.
\]
The implication is the residuum of the t--norm and it can be also defined as $x\rightarrow y:=\neg x\oplus y$. In addition, $\underbar{0}\equiv x\otimes\neg x$ and $\underbar{1}\equiv x\oplus\neg x$ for any $x\in\mathcal{V}$, where $\mathcal{V}$ denotes the set of all the propositional variables. It is worth to notice that in {\bf \L} we have two conjunctions and two disjunctions (a strong and a weak). Weak connectives can be defined from the strong ones in every fuzzy logic. For instance  the weak conjunction that corresponds to the minimum operation, can be defined as $x\wedge y:=x\otimes(x\rightarrow y)$. Even if the two conjunctions and the two disjunctions coincide on the crisp values $\{0,1\}$ they generalize quite differently on continuous values $[0,1]$.

We conclude this brief overview on \L ukasiewicz logic recalling some fundamental properties. For instance, the \emph{De Morgan laws} between both weak and strong connectives hold, as well as the \emph{Distributive laws} of strong over weak. We only notice that the distributive property does not hold between strong conjunction and strong disjunction, namely for instance
\[
x\oplus(y\otimes z)\not\equiv(x\oplus y)\otimes(x\oplus z),
\]
for some $x,y,z$, as shown in the following counter--example taking $x=0.1,\,y=z=0.5$:
\[
\min\{1,x+\max\{0,y+z-1\}\}=0.1
\]
\[
\max\{0,\min\{1,x+y\}+\min\{1,x+z\}-1\}=0.2.
\]

In view of the learning procedure, formulated as an optimization problem, we are interested in a \emph{functional representation} of logical formulas. Indeed, FOL will be translated into functional constraints for the objective functions, which are to be thought of as predicates or as any--arity relations. Since the algebra of {\bf \L}-formulas on $n$ variables is isomorphic to the algebra of functions from $[0,1]^n$ to $[0,1]$ \cite{novakmathematical}, we can translate formulas into functions in a very natural way. In particular, the zero formula $\underbar{0}$ corresponds to the constant function equal to 0 and, given $\circ\in\{\wedge,\vee,\otimes,\oplus,\rightarrow\}$, for every $(x_1,\ldots,x_n)\in[0,1]^n:$
\[
f_{\neg\varphi}(x_1,\ldots,x_n)=1-f_\varphi(x_1,\ldots,x_n),
\]
\[
f_{\varphi\circ\psi}(x_1,\ldots,x_n)=f_{\varphi}(x_1,\ldots,x_n)\circ^* f_{\psi}(x_1,\ldots,x_n);
\]
where $\circ^*$ denotes the operation corresponding to the connective $\circ$ according to Table \ref{tab:luksem}.

An optimization problem can be solved with efficient algorithms in the case it is convex, or even better if it is quadratic. In particular, convexity guarantees the existence (and unicity in the strong convex case) of an optimal solution under opportune hypothesis, as we will see in the next section. In addition, convex optimization is not affected by the presence of local minima, indeed any local optimum is also global. These are the main reasons why we decided to investigate the convexity of the functions corresponding to logical formulas. A complete functional representation for G$\ddot{\mbox{o}}$del, Product and \L ukasiewicz fuzzy logics can be found in \cite{aguzzoli2011free}. However, in this paper we are only interested in the latter, indeed functions corresponding to Product t--norm are at most \emph{quasi concave} and the G$\ddot{\mbox{o}}$del t--norm is represented by the \L ukasiewicz weak conjunction. The fundamental result about the functional representation of {\bf \L}--formulas is given by the well--known McNaughton Theorem \cite{novakmathematical}. It states that, for the propositional case, the algebra of formulas of \L ukasiewicz Logic on $n$ variables is isomorphic to the algebra of McNaughton functions defined on $[0,1]^n$.

\begin{definition}
	Let $f:[0,1]^n\rightarrow[0,1]$ be a continuous function with $n\geq0$, $f$ is called a McNaughton function if it is piecewise linear with integer coefficients, that is, there exists a finite set of linear polynomials ${p_1,\ldots,p_m}$ with integer coefficients such that for all $(x_1,\ldots,x_n)\in[0,1]^n$, there exists $i\leq m$ such that 
	\[
	f(x_1,\ldots,x_n)=p_i(x_1,\ldots,x_n).
	\]
\end{definition}

\begin{theorem}[McNaughton Theorem]
	For each $n\geq0$, the class of $[0,1]$-valued functions defined on $[0,1]^n$ that correspond to formulas of propositional \L ukasiewicz logic coincides with the class of McNaughton functions defined on $[0,1]^n$ and equipped with pointwise defined operations.
\end{theorem}

As a consequence, for every {\bf \L}--formula $\varphi$ depending on $n$ propositional variables, we can consider its corresponding function $f_\varphi:[0,1]^n\rightarrow[0,1]$, whose value on each point is exactly the evaluation of the formula with respect to the same variable assignment. Hence we can investigate the convexity of such functions that, for the McNaughton Theorem, are McNaughton functions.

\subsection{Convex McNaughton Functions}
As we said above, we are now interested in the as large as possible set of {\bf \L}--formulas corresponding to convex McNaughton functions. We start with the investigation of logical connectives corresponding to operations that preserve concavity or convexity\footnote{For an explicit definition of concave (convex) function, see Appendix \ref{appendix}.}.

\begin{lemma}\label{concpres}
	Let $\varphi,\,\psi$ be two {\bf \L}--formulas, then
	\begin{enumerate}
		\item $f_\varphi$ is convex if and only if $f_{\neg\varphi}$ is concave;
		\item if $f_\varphi,f_\psi$ are concave then the functions $f_{\varphi\wedge\psi}$ and $f_{\varphi\oplus\psi}$ are concave;
		\item if $f_\varphi,f_\psi$ are convex then the functions $f_{\varphi\vee\psi}$ and $f_{\varphi\otimes\psi}$ are convex.
	\end{enumerate}
\end{lemma}
\begin{proof}
	See Appendix.
\end{proof}
Therefore, the operations corresponding to the connectives $\wedge,\,\oplus$ preserve concavity, while that ones corresponding to $\vee,\,\otimes$ preserve convexity. In the following definition, we fix two different fragments of \L ukasiewicz formulas which are built according to such connectives.

\begin{definition}
	Let $(\wedge,\oplus)^*$ and $(\otimes,\vee)^*$ be the smallest sets of formulas (up to equivalence) such that:
	\begin{itemize}
		\item $\underbar{0}\in(\wedge,\oplus)^*$ and $\underbar{1}\in(\otimes,\vee)^*$;
		\item if $x\in\mathcal{V}$, then $x,\neg x\in(\wedge,\oplus)^*$ and $x,\neg x\in(\otimes,\vee)^*$;
		\item if $\varphi_1,\varphi_2\in(\wedge,\oplus)^*$, then $\varphi_1\wedge\varphi_2,\,\varphi_1\oplus\varphi_2\in(\wedge,\oplus)^*$;
		\item if $\varphi_1,\varphi_2\in(\otimes,\vee)^*$, then $\varphi_1\otimes\varphi_2,\,\varphi_1\vee\varphi_2\in(\otimes,\vee)^*$.
	\end{itemize}
\end{definition}
Anticipating the main result of this section, we refer to $(\wedge,\oplus)^*$ as \emph{the concave fragment} and to $(\otimes,\vee)^*$ as \emph{the convex fragment} of \L ukasiewicz logic. Since $\underbar{0},\underbar{1}$ correspond to constant functions and the literals correspond to projections or their negations, which are affine functions and therefore both concave and convex, the formulas inside a specific fragment are guaranteed to be concave or convex respectively. However, thanks to the following theorem (see e.g. \cite{rockafellar2009variational}), we are able to prove the if and only if claim.

\begin{theorem}
	Any convex piecewise linear function on $\mathbb{R}^n$ can be expressed as a max of a finite number of affine functions.
\end{theorem}
This means that, for each $n$-ary convex McNaughton function $f$ there exist $M_1,\ldots,M_k\in\mathbb{Z}^n,\,q_1,\ldots,q_k\in\mathbb{Z}$ such that:
\begin{equation}\label{convlin}
\mbox{for all }x\in[0,1]^n\qquad f(x)=\max_{i=1,\cdots,k}\,\left(M_i'\cdot x+q_i\right).
\end{equation}
On the other hand, every concave McNaughton function can be expressed as the minimum of a finite number of affine functions. We only mention that the coefficients of the affine functions are constructively determined by the shape of the considered formula. 

{\em Example.} Let us consider $\varphi=((x\wedge y)\oplus\neg y\oplus z)\wedge \neg z$, then $\varphi\in(\wedge,\oplus)^*$ and it is equivalent to $(x\oplus \neg y\oplus z)\wedge(y\oplus\neg y\oplus z)\wedge\neg z$. From this latter expression, we get:
\[
f_\varphi(x,y,z)=\min\{1,x-y+z+1,1-z\}.
\]
Finally, exploiting (\ref{convlin}) and thanks to Lemma \ref{concpres}, we can prove that $(\wedge,\oplus)^*$ and $(\otimes,\vee)^*$ coincide with the sets of all {\bf \L}-formulas whose corresponding McNaughton functions are concave and convex respectively.

\begin{proposition}\label{prop:conc}
	Let $f_\varphi:[0,1]^n\rightarrow[0,1]$ be a McNaughton function. Then,
	\begin{enumerate}
		\item $f_\varphi$ is concave if and only if  $\varphi\in(\wedge,\oplus)^*$;
		\item $f_\varphi$ is convex if and only if  $\varphi\in(\otimes,\vee)^*$.
	\end{enumerate}
\end{proposition}
\begin{proof}
	See Appendix \ref{appendix}.
\end{proof}
As a result, the largest fragments of {\bf \L} whose McNaughton functions are either concave or convex are determined. It is worth to notice that in general, in literature logical constraints are often initially expressed in boolean form and then they are \emph{fuzzified}. Since the fragments we define contain both a conjunction and a disjunction which are coherent with boolean connectives on the crisp values, in principle one can almost always translate boolean formulas into convex \L ukasiewicz constraints. However, as we sketch in the next section, any choice for this translation can slightly modify the expressiveness of the formulas we were dealing with.

\subsection{Notes on Expressiveness}\label{sec:expressiv}
In the whole paper, we suppose to deal with a prior knowledge expressed by a set of first--order {\bf \L}--formulas. However in the majority of cases, in practical problems we are given a set of boolean formulas, hence we have to decide a suitable way of translating them into the language of \L ukasiewicz logic. The majority of authors \cite{serafini2016logic,diligenti2017semantic,brocheler2012probabilistic} convert boolean conjunction with the t--norm and the disjunction with the t--conorm, namely with the pair $(\otimes,\oplus)$. It is worth noticing that Proposition \ref{prop:conc} characterizes the concave and the convex formulas in {\bf \L}, but in general it does not provide an effective way to embed any boolean formula into a specific fragment. However, this is always guaranteed by making use of either the concave or the convex connectives if we consider any boolean formula without implication and with negation on at most literals. In particular, given any boolean formula, we can always rewrite it into a \emph{conjunctive normal form} (CNF) and then we can apply the following two translations to the conjunctions and disjunctions to fall into the concave or the convex fragment respectively:
\begin{equation}\label{eq:normal}
\mbox{(concave)}\;\;\bigwedge_{i=1}^n\bigoplus_{j=1}^m(\neg)l_{ij},\quad\mbox{(convex)}\;\;\bigotimes_{i=1}^n\bigvee_{j=1}^m(\neg)l_{ij}.
\end{equation}

Actually, there are several possibilities to translate the boolean connectives into the \L ukasiewicz ones, where we have two conjunctions and two disjunctions. Even if the weak and the strong operations coincide on the boolean values $\{0,1\}$, the way they generalize on continuous values $[0,1]$ determines different semantics for the resulting formulas. In the following we show some examples to compare possible representations.

\begin{example}[Distributivity]
	In general, we can lose consistency on what we have to represent if we manipulate the boolean expressions before translating them into fuzzy terms. For instance, in boolean logic, the two formulas $x \wedge (y\vee z)$ and  $(x\wedge y)\vee( x\wedge z)$ are equivalent. However, if we translate them with the fragment $(\otimes,\vee)^*$ the equivalence still holds, whereas in general the same is not true with the fragment $(\otimes,\oplus)^*$. 
\end{example}

A well--studied class of formulas is given by the Horn clauses, i.e. formulas with a propositional variable implied by a conjunction of propositional variables. They are very common in the literature since they are quite expressive and easy to handle.

\begin{example}[Horn Clauses]
	Given a Horn clause, if we translate the conjunction with the t--norm, then we get:
	\[
	(x_1\otimes\ldots\otimes x_m)\rightarrow y\equiv(\neg x_1\oplus\ldots\oplus\neg x_m\oplus y)\in(\wedge,\oplus)^*,
	\]
	namely the class of Horn clauses translated in {\bf \L} with $\otimes$ is strictly contained in the concave fragment. Indeed, the following formulas are in $(\wedge,\oplus)^*$ and do not correspond to Horn clauses: 
	\[
	x\wedge y,\qquad x\oplus(y\wedge z),\qquad x\wedge(x\rightarrow y),\qquad (x\otimes y)\rightarrow(x\wedge z).
	\]
\end{example}

Finally in the rest of this section, we discuss some examples representing well--known rules in learning from logical constraints.

\begin{example}[Manifold Regularization]
	Manifold regularization assumes that a given relation $R(a,b)$ states when two objects $a$ and $b$ belong to the same manifold,
	requiring that the value of a predicate $P(x)$ should be consistent when evaluated on the two points. This condition can be expressed
	by the following boolean formula:
	\begin{equation}\label{mr}
	R(a,b)\rightarrow(P(a)\leftrightarrow P(b)),
	\end{equation}
	where $P(a)\leftrightarrow P(b)\equiv (P(a)\rightarrow P(b))\wedge(P(b)\rightarrow P(a))$.
	\\
	
	Since for all $x,y\in[0,1]$
	\[
	\min\{1,1-x+y,1-y+x\}=\max\{0,\min\{1-x+y,1-y+x\}\}
	\]
	then, $(x\rightarrow y)\wedge(y\rightarrow x)\equiv(x\rightarrow y)\otimes(y\rightarrow x)$ and in this case each choice we make between $\wedge$ and $\otimes$ yields an equivalent results. In particular, if we use the weak conjunction we can immediately see that such a formula belongs to the concave fragment.
\end{example}

\begin{example}[Mutually Exclusive Classes]
	One can ask, for instance in a collective classification problem, that a certain pattern belongs to one and only one of two (or more) classes. For short we indicate with $x$ and $y$ the two grounded predicates corresponding to the class assigned to the same object $a$. For instance, by means of $x\vee y$ we can express that $a$ belongs to at least one of the two classes and by $(x\wedge\neg y)\vee(\neg x\wedge y)$ that $a$ belongs to exactly one class. Such formulas can be translated into \L ukasiewicz logic in different ways. However it seems that some choices are more accurate with respect to the initial boolean semantics of the formula. For instance, if we translate $x\vee y$ with $x\oplus y$, then the formula will be satisfied for any pair of $[0,1]$--values summing to 1, on the other hand it can be more useful to translate it with the weak disjunction that corresponds to the maximum. For what concerns the exclusive disjunction, it can be represented in \L ukasiewicz logic by $(x\otimes\neg y)\vee(\neg x\otimes y)$ that belongs to the convex fragment.
	
\end{example}

\section{Learning Schemes}
\label{sec:zoo}
In this section we briefly introduce some learning frameworks where we are able to get some theoretical insights when exploiting the convex \L ukasiewicz fragment. In the big picture, we are interested in the learning of a set of real--valued functions ${\bf P}=\{p_j:\mathbb{R}^{n_j}\rightarrow\mathbb{R},\,j\in\mathbb{N}_J,\,n_j>0\}$ provided with some factual (supervisions) and abstract knowledge (logical formulas) on them. Throughout the paper, each objective function $p_j$ is supposed to be evaluated on a supervised training set $\mathscr{L}_j$ and an unsupervised training set $\mathscr{U}_j$, where:
\[
\mathscr{L}_j=\{({\bf x}_l^j,y_l^j):\,{\bf x}_l^j\in\mathbb{R}^{n_j},y_l^j\in\{-1,1\},\,l\in\mathbb{N}_{l_j}\},
\]
\[
\mathscr{U}_j=\{{\bf x}_u^j\in\mathbb{R}^{n_j}:\,u\in\mathbb{N}_{u_j}\}.
\]
For any $j$, the logical constraints related to the $j$-th predicate will be forced on the set $\mathscr{U}_j$. Since any quantifier is applied to a specific argument of a predicate, it can be useful to decompose by components the range of vectors ${\bf x}_u^j$ as $\mathscr{U}_j=\mathscr{U}_{j1}\times\ldots\times\mathscr{U}_{jn_j}$ such that each $\mathscr{U}_{jk}$ is the domain of the $k$-th argument of the predicate $p_j$.
It is also useful to define, the set $\mathscr{S}_j=\mathscr{U}_j\cup\mathscr{L}_j'$ (where $\mathscr{L}_j'=\{{\bf x}_l^j:\,({\bf x}_l^j,y_l^j)\in\mathscr{L}_j\}$) that contains the whole set of points in which the $j$-th objective is evaluated. In the following, we indicate with $s_j$ the cardinality of $\mathscr{S}_j$, $S=s_1+\ldots+s_J$ and $U=u_1+\ldots+u_J$. Further, we suppose we are given a set of \L ukasiewicz FOL formulas,  $KB=\{\varphi_h:\,h\in\mathbb{N}_H\}$ whose predicates are functions in {\bf P}. With respect to the notation, in the following we adopt some abbreviations. For instance, will omit the superscript $j$ on points and we write $p_{jl}$ instead of $p_j({\bf x}_l^j)$. In addition, we write ${\bf p}=(p_1,\ldots,p_J):\mathbb{R}^n\rightarrow\mathbb{R}^J$, where $n=n_1+\ldots+n_J$, for the overall objective function.


\subsection{Kernel Machines}
Kernel methods are a class of algorithms for pattern analysis that exploit high--dimensional feature representation. Among others, one of such algorithms that is widely employed is the well--known \emph{Support Vector Machine} (SVM), a supervised learning model aiming to separate a set of points that belong to different classes with an as large as possible margin \cite{boser1992training}. One of the goals of SVMs is that learning can be efficiently solved by quadratic optimization algorithms both in the primal and in the dual space. In the following, we show how to extend its classical formulation to include logical constraints still preserving quadratic programming. From now on in this section, we assume that each objective function $p_j$ belongs to some \emph{Reproducing Kernel Hilbert Space} (RKHS) $\mathcal{H}_j$. By means of the reproducing property, $p_j$ can be represented as an expansion of the kernel function $k_j\in\mathcal{H}_j$. The choice of the kernel is empirically validated, typical choices are  \emph{polynomial} or a \emph{gaussian} kernel.
\\

\subsubsection*{Constraints}
For every objective function $p_j$, we require the satisfaction of soft constraints corresponding to supervisions in $\mathscr{L}_j$ according to the classical approach with hinge loss functions:
\[
y_l(2 p_j({\bf x}_l)-1)\geq 1-2\xi_{j_l} \mbox{ with }\xi_{j_l}\geq0, \mbox{ for }l=1,\ldots,l_j,
\]
where the slack variable $\xi_{j_l}$ denotes the smallest nonnegative number satisfying the constraint. In the following, we refer to such constraints as \emph{pointwise constraints}. They are slightly modified with respect to the usual formulation, since the supervisions are labelled with values $y_l=\pm1$ whereas the predicates are supposed to assume 0--1 values denoting classical logic true--values. 

The objective functions in {\bf P} must assume $[0,1]$-values in all the points on which they are evaluated. Indeed they have to behave as logical predicates occurring in the formulas of KB. In previous works (e.g. \cite{diligenti2012bridging}) this is achieved applying a sigmoid function to predicates before enforcing the logical constraints. Since in general the sigmoid does not preserve convexity (as well as concavity), we opt for a different strategy, by requiring explicitly the 0--1 bound for every $p_j$. These linear constraints are referred to as \emph{consistency constraints}:
\begin{equation}\label{cc}
0\leq p_j({\bf x}_s)\leq 1, \mbox{ for } s=1,\ldots,s_j,\, {\bf x}_s\in\mathscr{S}_j.
\end{equation}

Finally, the \emph{logical constraints} arise from the knowledge base $KB$ that is supposed to be a collection of FOL {\bf \L}--formulas. However, according to (\ref{qq}) each quantified formula can be replaced with a propositional one once all the predicates are grounded on their domains of evaluation. We denote the set of such \emph{propositionalized} formulas, where the grounded predicates are considered as $[0,1]$ propositional variables, by $KB'$. For what concerns logical constraints, every $p_j$ is supposed to be evaluated on the unsupervised examples in $\mathscr{U}_j$. Since the formulas in $KB'$ depend, in general, on all the possible groundings of their occurring predicates, it is useful for $j=1,\ldots,J$ to indicate with ${\bf \bar{p}}_j$ the vector of all groundings of $p_j$ in $\mathscr{U}_j$, namely ${\bf \bar{p}}_j=(p_{j1},\ldots,p_{ju_j})$ and ${\bf \bar{p}}=({\bf \bar{p}}_1,\ldots,{\bf \bar{p}}_J)\in[0,1]^U$.

\begin{example}
	Let us consider ${\bf P}=\{p_1,p_2\}$ and $\varphi\in KB$ such that:
	\[
	\varphi:\,\forall x\exists y(p_1(x)\rightarrow p_2(x,y)).
	\]
	Given $\mathscr{U}_{11}=\mathscr{U}_{21}=\{x_1,x_2\}$ and $\mathscr{U}_{22}=\{y_1,y_2\}$, we have $\mathscr{U}_1=\mathscr{U}_{11}$ and $\mathscr{U}_2=\{(x_1,y_1),(x_1,y_2),(x_2,y_1),(x_2,y_2)\}$, hence the grounding vectors for the two predicates are
	\[
	{\bf \bar{p}}_1 = \left(p_{11}, p_{12}\right), \  \ {\bf \bar{p}}_2=\left(p_{21}, p_{22}, p_{23}, p_{24}\right).
	\]
	Therefore, we get as propositional form for $\varphi$ the formula
	\[		
	\left[(p_{11}\rightarrow p_{21})\vee(p_{11}\rightarrow p_{22})\right]\wedge\left[(p_{12}\rightarrow p_{23})\vee(p_{12}\rightarrow p_{24})\right] \ .
	\]	
\end{example}

Since we are now dealing with propositional formulas, all the results reported in the previous section can be used. In particular, every $n-$ary formula $\varphi$ in $KB'$ is isomorphic to a McNaughton function $f_\varphi:[0,1]^n\rightarrow[0,1]$. For the sake of simplicity, $f_h$ indicates the function corresponding to the formula $\varphi_h$ and $KB'$ is the set of such functions: $KB'=\{f_1,\ldots,f_H\}$. In addition, to make the notation as uniform as possible, we write any of such function as depending on the grounding vector of all predicates, i.e. $f_h=f_h({\bf \bar{p}})$, however in general it depends on only few predicates. Finally, we can express the logical constraints, and requiring their soft satisfaction, a new slack variable is introduced for each $h\in\mathbb{N}_H$, such as
$1-f_h({\bf \bar{p}})\leq\xi_h\mbox{ with }\xi_h\geq0$.
If $KB'\subseteq(\wedge,\oplus)^*$, namely if the formulas are built from the concave fragment, then $f_h$ is a concave function and $1-f_h$ is the convex function corresponding to the formula $\neg\varphi_{h}$.

The considered McNaughton functions are convex in the space of grounded predicates. Since we suppose that each objective function belongs to a certain RKHS, then we can write $p_j({\bf x})=\omega_j'\cdot\phi_j({\bf x})+b_j$, where $\phi_j:\mathbb{R}^{n_j}\rightarrow\mathbb{R}^{N_j}$ is a feature map determined by the $j$-th kernel function $k_j$ of $\mathcal{H}_j$, 	$\omega_j\in\mathbb{R}^{N_j}$ is said the $j$-th weight vector and $b_j\in\mathbb{R}$. If we assume ${\bf x}\in\mathscr{U}_j$ and we set $\hat{\omega}_j=(\omega_j',b_j)'$, then the values of predicates are totally determined by the matrix $\hat{\omega}=(\hat{\omega}_1,\ldots,\hat{\omega}_J)$. This entails that the formulas will be evaluated by composition on the weight space and therefore we need to guarantee the convexity of McNaughton functions on this space. Thanks to the linear form assumed by each objective function in the feature space (in general $N_j>>n_j$ or even $N_j$ can be infinite), the following lemma applies and guarantees convexity of the functional logical constraints in the weight space too.

\begin{lemma}
	Let $f:Y\subseteq\mathbb{R}^{m}\rightarrow\mathbb{R}$ be a convex (concave) function and $g:X\subseteq\mathbb{R}^{d}\rightarrow Y$ such that $g(x)=Ax +b$ with $A\in\mathbb{R}^{m,d},b\in\mathbb{R}^m$. Then the function $h:X\rightarrow\mathbb{R}$ defined by $h=f\circ g$ is convex  (concave) in $X$.
\end{lemma} 

Summing up, we can embed the logical constraints into the overall loss function by means of convex functional constraints if we consider only formulas that belong to the concave fragment. However, since we suppose to deal with \L ukasiewicz formulas, any functional constraint is a McNaughton function, and a piecewise linear function, in particular. Therefore according to $(\ref{convlin})$ we have for every $h\in\mathbb{N}_H$: 
\[
1-f_h({\bf \bar{p}})=\max_{i=1,\cdots,I_h}\left(M^h_i\cdot {\bf \bar{p}}+q^h_i\right)\leq\xi_h
\]
if and only if
\begin{equation}\label{cl}
M^h_i\cdot {\bf \bar{p}}+q^h_i\leq\xi_h\quad\mbox{ for all }i\in\mathbb{N}_{I_h},
\end{equation}
where $M^h_i\in\mathbb{R}^{1,U}$ and $q^h_i\in\mathbb{R}$ are integer coefficients determined by the shape of the formula $\neg\varphi'_h$. This means that we can replace each convex constraint with a set of linear constraints (also in the weight space).

It is worth to notice that, since a distributive law holds in any fragment for the strong connective with respect to the weak one, we can easily rewrite any concave formula as a weak conjunction of strong disjunctions as well as any convex formula as a weak disjunction of strong conjunctions. Thereafter is straightforward to get the integer coefficients of the affine functions. For instance, given $\varphi\in(\wedge,\oplus)^*$ such that $\varphi:\,(x_1\oplus\neg x_2)\wedge(x_1\oplus x_2)$, we have:
\[
f_\varphi=\min\{1,x_1-x_2+1,x_1+x_2\}.
\]

\subsubsection*{Optimization}
In order to learn the objective functions we have to find the optimal values for the weight matrix. As in SVMs we ask for a solution that maximizes the margin between the false and the true class satisfying the constraints. Functional logical constraints are expressed by their linear counterparts, such that we can formulate a quadratic optimization problem in the primal space as well as in the dual space. Here we report the \emph{primal} formulation of the problem together with some comments on the optimal solution.  We refer to the Appendix \ref{appendix} for further details. 
\begin{equation}\label{pp}
\mbox{{\bf Primal Problem}}
\end{equation}
\[
\min\; \frac{1}{2}\sum_{j\in\mathbb{N}_J}||\omega_j||^2+C_1\sum_{\substack{j\in\mathbb{N}_J\\ l\in\mathbb{N}_{l_j}}}\xi_{j_l} + C_2\sum_{h\in\mathbb{N}_H}\xi_h \quad\mbox{subject to:}
\]
\[
\begin{array}{ll}
y_l(2 p_j({\bf x}_l)-1)\geq 1-2\xi_{j_l}, & \xi_{j_l}\geq0,\\
M^h_i\cdot {\bf \bar{p}}+q^h_i\leq\xi_h, & \xi_h\geq0,\\
0\leq p_j({\bf x}_s)\leq 1, & 
\end{array}
\]
where $j\in\mathbb{N}_J,\,l\in\mathbb{N}_{l_j},\,({\bf x}_l,y_l)\in\mathscr{L}_j$, $h\in\mathbb{N}_H,\,i\in\mathbb{N}_{I_h}$, $s\in\mathbb{N}_{s_j},\,{\bf x}_s\in\mathscr{S}_j$ and $C_1,C_2$ are positive real parameters determined by cross validation.

\begin{remark}
	The constants $C_1$ and $C_2$ express the (possibly different) degree of satisfaction for the pointwise and logical constraints respectively. It is worth noticing that the supervisions can be seen as atomic logical constraints or their negation. However we decided to keep them separated in this formulation both for clarity with respect to the usual SVM literature and for considering different values for the constants. In principle one can weigh any constraint differently. However we suppose to have the same degree of belief on all the supervisions as well as on all the logical formulas.
\end{remark}

The problem can be solved in the primal or dual space, because convexity guarantees that the KKT--conditions are also sufficient and the duality gap is null.
Indeed,  The solution of the $j$-th objective with respect to the optimal parameters can be written as:
\[
p_j({\bf x})=\omega^*_j\cdot\phi_j({\bf x})+b^*_j=2\sum_{l=1}^{l_j}\lambda^*_{j_l}y_lk_j({\bf x}_l,{\bf x})+
\]
\[
-\sum_{h=1}^H\sum_{i=1}^{I_h}\lambda^*_{h_i}\sum_{u=1}^{u_j}M^h_{i,u}k_j({\bf x}_u,{\bf x})+\sum_{s=1}^{s_j}(\eta^*_{j_s}-\bar{\eta}^*_{j_s})k_j({\bf x}_s,{\bf x})+b^*_j.
\]

\subsection{Collective Classification}
In general the logical constraints express relational informations on different objective functions at a time, hence they can be very suitable for \emph{collective} classification tasks. Here we formulate the collective setting in the same spirit as what we did for kernel machines. The logical formulas are collected in $KB$ and each predicate $p_j$ is supposed to be evaluated on its sample set $\mathscr{S}_j$, while ${\bf \bar{p}}_j=(p_{j1},\ldots,p_{js_j})$ denotes the vector of all its possible groundings, ${\bf \bar{p}}=({\bf \bar{p}}_1,\ldots,{\bf \bar{p}}_J)$.

In the previous subsection, we assumed to learn the objective functions by the training of opportune kernel machines, whereas in this case, we assume that an appropriate model (e.g. a neural network) has already been trained to compute their prior values. In the following, we indicate by ${\bf \hat{p}}_{j}$ the available vector of priors for the $j$-th objective function. Even if the priors assume $[0,1]$--values, we have to guarantee the same for the values into the grounding vectors of the objective functions. Then, we enforce again (\ref{cc}).

Given this setting, the considered multi--task learning problem is formulated as	
\[
\min_{\bf p} \; \sum_{j\in\mathbb{N}_J}||{\bf \bar{p}}_j - {\bf \hat{p}}_j||^2+C_1 \sum_{h\in\mathbb{N}_H}\xi_h \quad\mbox{subject to}
\]
\[
\begin{array}{ll}
1-f_h({\bf \bar{p}})\leq\xi_h, & \xi_h\geq0,\\
0\leq p_j({\bf x}_s)\leq 1 & 
\end{array}
\]
where $h\in\mathbb{N}_H$, $j\in\mathbb{N}_J,\,s\in\mathbb{N}_{s_j},\,{\bf x}_s\in\mathscr{S}_j$ and $C_1$ is used to weigh the degree of satisfaction of the logical constraints. This optimization problem aims at finding the grounding for the predicates closest to the given priors and yielding the minimal violation of the constraints. If we assume to deal with formulas in the concave fragment, then the logical constraints can be rewritten according to (\ref{cl}) and we obtain a quadratic programming problem that can be efficiently solved.

\subsubsection*{Manifold Regularization}
As an example, we discuss here the already mentioned logical rule for manifold regularization (\ref{mr}). In principle, given any binary relation $R(x,y)$ on a domain of a predicate, we can require that this predicate assumes as close as possible values on those points that are related by $R$. For instance, the topological properties of the original domains of the predicates are not explicitly represented in the considered setting, apart from the values assigned for the given priors. This suggest to consider a predefined binary relation $R({\bf x}_1,{\bf x}_2)$ expressing the membership of the two points to a same manifold. For instance, a spatial regularization manifold can be defined as $R({\bf x}_1,{\bf x}_2)=\exp\left(-\frac{||{\bf x}_1-{\bf x}_2||^2}{\sigma^2}\right)$,
where ${\bf x}_1,{\bf x}_2\in\mathscr{S}_j$ for some $j\in\mathbb{N}_J$, and with $\sigma$ as neighbourhood width parameter. To enforce the manifold regularization we enforce the satisfaction of the following logical formula for each ${\bf x}_1,{\bf x}_2\in\mathscr{S}_j$,
\[
R_{j12} \rightarrow ((p_{j1} \rightarrow p_{j2})\wedge(p_{j2} \rightarrow p_{j1}))
\]
that is equivalent to the convex constraint 
\[
\max\{0,\,R_{j12}+p_{j1}-p_{j2}-1,\,R_{j12}-p_{j1}+p_{j2}-1\}\leq\xi,
\]
and to the following linear constraints
\[
\begin{array}{l}
R_{j12}+p_{j1}-p_{j2}-1\leq\xi,\\
R_{j12}-p_{j1}+p_{j2}-1\leq\xi.
\end{array}
\]

%
%

\subsection{Probabilistic Soft Logic}
Probabilistic soft logic (PSL) \cite{kimmig2012short,bach2015hinge} is a general framework for probabilistic reasoning in relational domains. Similarly to Markov Logic Networks \cite{richardson2006markov}, PSL uses first--order logic rules to instantiate a graphical model having as nodes the values of each grounded predicate, represented as soft--assignments in $[0,1]$.

PSL uses the \L ukasiewicz logic to implement a relaxation technique commonly used to solve MAX SAT problems. In particular, let $C = \{c_1, \ldots , c_m\}$ be a set of logic disjunctive clauses, where each formula in the disjunction is a literal, i.e. an atomic formula or its negation. 
PSL embeds the knowledge into a Markov Random Field (MRF), which builds a distribution over possible interpretations as:
\[
P(\mathcal{I}) = \frac{1}{Z} \exp \left(-\sum_{j=1}^m \lambda_j~ \Phi_j(\mathcal{I}) \right)
\]
where $\lambda_j \ge 0$ is the weight of the clause $c_j$, $Z$ is a normalization constant and the potential $\Phi_j$ expresses the distance from the satisfaction of the formula $c_j$. Each weight $\lambda_j$ can be used to express how strongly the $j$-th clause is enforced to hold true. In fact, a higher weight penalizes stronger an assignment that does not satisfy the corresponding clause.

PSL assumes the assignment of a template to each clause $c_j$, that is reused for each single grounding of the clause in the interpretation. Assuming that a clause is universally quantified, the MRF has one clique for each grounding of such formula and the potential $\Phi_j$ can be expressed as the sum of the potential $\phi_j$ on all the possible groundings. In particular, $\Phi_j (\mathcal{I}_j) = \sum_{g \in \mathcal{I}_j} \phi_j(g)$, where
$\mathcal{I}_j$ is the set of groundings of the $j$-th formula with respect to the interpretation $\mathcal{I}$. Let $I^{+}_j$ and $I^{-}_j$ be the indexes of the positive and negative literals respectively in a grounding of $c_j$. PSL employs the \L ukasiewicz logic to express $\phi_j$ and since disjunctive clauses are supposed, $\phi_j$ results into the following convex functional:
\begin{equation}\label{eq:psl}
\phi_j(g) = \max\{0,1-\sum_{k \in I^{+}_j} g_k-\sum_{k \in I^{-}_j} (1 - g_k)\}.
\end{equation}

The PSL framework defines an efficient method to perform inference and to determine the most likely interpretation given the available evidence. This is equivalent to minimize the summation in the exponential, which corresponds to a linear (convex) optimization problem under the restrictions defined above. However, using the concave \L ukasiewicz fragment proposed in this paper, inference remains tractable also when lifting the restriction to disjunctive formulas. In fact, as we have already mentioned, any boolean formula can be rewritten in (CNF) and then embedded into the concave fragment exploiting (\ref{eq:normal}). However, the negation of such formula, that expresses its distance from the satisfaction to be minimized, can be written as the convex functional:
\[
\max\{0,1-l_1(g),\ldots,1-l_n(g)\}
\]
where for $i=1,\ldots,n$, $l_i(g)$ corresponds to the grounding $g$ of some logic disjunctive clause $l_i$. Since the expression above, can be thought of as a potential that extends (\ref{eq:psl}) still preserving convexity, we can extend the set of formulas represented in PSL to the whole concave fragment $(\wedge,\oplus)^*$.

\subsubsection*{Weight Learning}
In the learning formulations of this paper, i.e. (\ref{pp}), we assumed for simplicity to have no preferences among the logical rules, validating a shared parameter $C_2$ for their degree of satisfaction. However, PSL weight learning is commonly performed by maximizing the likelihood of the training data via gradient descent, where the derivative with respect to a weight $\lambda_j$ is:
$$
\frac{\partial \log P(\mathcal{I})}{\partial \lambda_j} = E_{\bf \lambda} \left[ \Phi_j(\mathcal{I}) \right] -\Phi_j(\mathcal{I}).
$$
The gradient with respect to the $j$--th clause weight is null when the distance from satisfaction of the training data $I_t$ corresponds to what is predicted by the model: $\Phi_j(\mathcal{I}_t) = E_{\bf \lambda}\left[ \Phi_j(\mathcal{I}) \right]$.
Computing the expected value is intractable in a general setting and improving PSL weight learning is an open research problem. A common solution is to approximate the expected value with the most probable interpretation according to the current weights:
$E_{\bf \lambda}\left[ \Phi_j(\mathcal{I}) \right] \approx \Phi_j(\mathcal{I}^\star)$.
Under this assumption, weight learning becomes tractable, because inference to determine the most probable interpretation corresponds to solving the convex optimization task as previously described.

\section{Experimental Results}
\label{sec:exp}
In this section we illustrate by means of some evaluations on an artificial dataset how to exploit our main result in practice. In the example, we consider the case of a given boolean formula which has to be translated into a continuous logical constraint. As we pointed out in Section \ref{sec:expressiv}, weak and strong connectives can behave quite differently still maintaining the coherence on the crisp values $0,1$.


\subsubsection*{Fuzzifications}
We consider the case of four predicates $A,B,C,D$ defined on the same domain $\mathcal{U}=[-3,3]\times[-3,3]$. Such predicates have to be thought of as membership functions of certain classes of patterns. In this case, we assume ${\bf A}=[-3,1]\times[-2,2]$, ${\bf B}=[-1,3]\times[-1,1]$, ${\bf C}=[-1,1]\times[-3,3]$, ${\bf D}=[-1,1]\times[-1,1]$.
Among the possible relations holding for the predicates corresponding to such classes, we decided to study the effect of two different continuous translations into \L ukasiewicz logic of the following boolean formula:
\begin{equation}\label{eq:constr}
\forall x:\,(A(x)\wedge B(x))\rightarrow (C(x)\wedge D(x)).
\end{equation}
As pointed out on the expressiveness notes, we can rewrite such a formula in CNF as
\[
(\neg A(x)\vee\neg B(x)\vee C(x))\wedge(\neg A(x)\vee\neg B(x)\vee D(x)),
\]
and, for instance, translate it in {\bf \L} with the mapping of the boolean connectives $(\wedge,\vee)$ either in ${\bf (1)}\,(\otimes,\oplus)$ that are the t--norm and the t--conorm or in ${\bf (2)}\, (\wedge,\oplus)$, that correspond to the convex operations. The soft constraints corresponding to such a formula in the two cases are:
\[
\begin{array}{crl}
{\bf (1)} & \min\{1,& \max\{0,A_i+B_i-C_i-1\}+ \\
&             &+\max\{0,A_i+B_i-D_i-1\}\}\leq\xi,\\
\\
{\bf (2)} & \max\{0,& A_i+B_i-C_i-1,A_i+B_i-D_i-1\}\leq\xi,
\end{array}
\]
where the subscript $i$ denotes the grounding of the corresponding predicate on the $i$-th point of the dataset.

For the evaluation we exploit a grid of points for any class. However only some subsets of these samples are provided with supervisions. In this example, we assume the sets {\bf A} and {\bf B} as fully labeled, while {\bf C} is partially labeled and {\bf D} is totally unsupervised. We compare the results obtained for the classes {\bf C} and {\bf D} when using the translations of the boolean formulas into \L ukasiewicz connectives, with respect to the percentage of supervisions on {\bf C}. We exploited a SVM, as described in (\ref{pp}), with a Gaussian kernel per predicate with standard deviation $\sigma=1$ and constant values $C_1=15$, $C_2=10$ for the constraints. The experiments are performed in MATLAB exploiting the interior-point algorithm with different runs over random permutations of the available labels for {\bf C}, while increasing the amount of supervisions.

The plots in figure \ref{fig:reca1v} report the means of the F1 score of the model without logical constraints, with translation {\bf (1)} (i.e. with no convexity) and with the convex fragment {\bf (2)}, respectively in the case we randomly initialize the optimization algorithms.
\begin{figure}[htbp]	
	\begin{center}
		\includegraphics[width=1\linewidth]{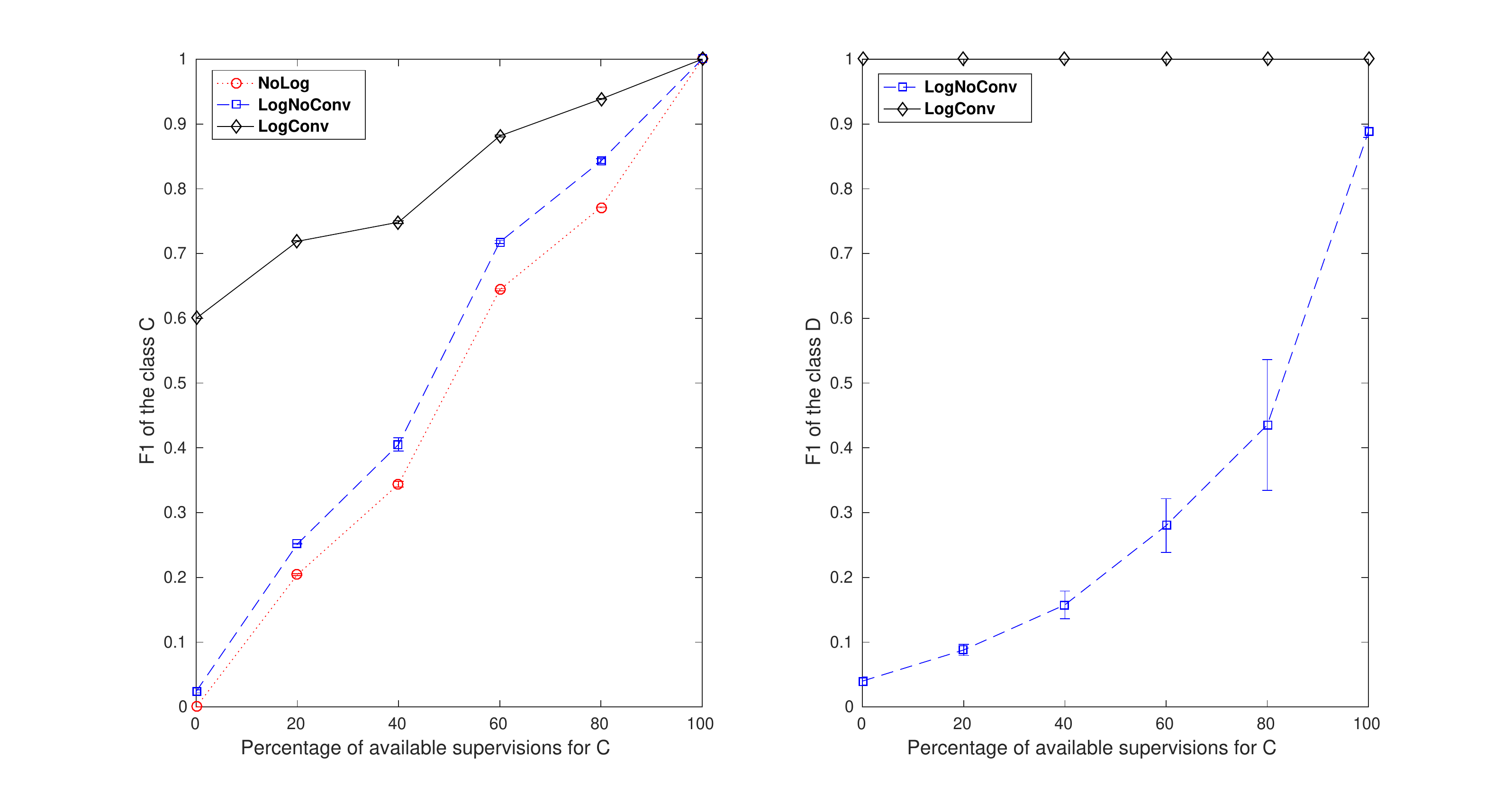}
	\end{center}
	\caption{F1 score of the classifiers on the classes C e D respectively, varying the amount of available supervisions in C.}
	\label{fig:reca1v}
\end{figure}

As we can see, the model with the convex constraint outperforms the other one given any number of supervisions for {\bf C}. However, the distance decreases when this number becomes bigger. For what concerns the class {\bf D}, if we increase the available supervisions of {\bf C} the recall for the convex case does not change, while we can observe some improvements for the non-convex classifier. Indeed the violation of the non convex constraint is due to the sum of two different contributes concerning both the classes. If we increase the number of  supervisions, then the contribution for {\bf C} to the constraint decreases and the optimization becomes more effective on the class {\bf D}.


\section{Conclusions}
\label{sec:conc}
%

The main contribution of this paper is a theoretical result about \L ukasiewicz logic that can be exploited in different learning schemes. In principle, whenever we are given a set of logical constraints, we can translate them into an equivalent form with only conjunctions, disjunctions and negations on propositional variables. Then if we translate them by the convex fragment of {\bf \L}, the constraints turn out to be convex and also equivalent to a set of linear constraints. Throughout the paper, we show how the provided theoretical result can be included into different learning settings in order to formulate a convex or even a quadratic optimization problem. In particular, we considered the integration into learning kernel machines, collective classification and Probabilistic Soft Logic.

\appendices
\section{}\label{appendix}
\begin{proof}[{\bf Derivation of the Dual Problem for SVMs}]
	From the problem (\ref{pp}), we can derive the Lagrangian function $\mathcal{L}(\hat{\omega},\xi,\lambda,\mu,\eta)$ as:
	\[
	\begin{array}{l}
	\displaystyle\frac{1}{2}\sum_j||\omega_j||^2+C_1\sum_{j,l}\xi_{j_l}+C_2\sum_h\xi_h-\sum_{j,l}\mu_{j_l}\xi_{j_l}+\\
	\displaystyle-\sum_{j,l}\lambda_{j_l}(y_l(2 p_j({\bf 	 x}_l)-1)-1+2\xi_{j_l})-\sum_h\mu_h\xi_h+\\
	\displaystyle-\sum_{h,i}\lambda_{h_i}(\xi_h-M^h_i\cdot {\bf \bar{p}}-q^h_i)-\sum_{j,s}((\eta_{j_s}-\bar{\eta})p_j({\bf x}_s)+\bar{\eta}_{j_s}).
	\end{array}
	\]
	If we set for all $j\in\mathbb{N}_J,l\in\mathbb{N}_{l_j},h\in\mathbb{N}_H,i\in\mathbb{N}_{I_h},s\in\mathbb{N}_{s_j}$, the usual {\bf KKT}--conditions, then the existence of the solution to the problem is guaranteed. In addition, by setting the null gradient condition of $\mathcal{L}$ with respect to  $\omega_j,b_j,\xi_{j_l},\xi_h$, we are able to formulate the problem also in the dual space:
	\[
	\max\; \theta(\lambda,\eta)=\mathcal{L}(\hat{\omega}^*,\xi^*,\lambda,\mu,\eta)\quad\mbox{subject to:}
	\]
	\[
	\begin{array}{l}
	\displaystyle\sum_{h,i}\lambda_{h_i}\sum_uM^h_{i,u} = 2\sum_l\lambda_{j_l}y_l+\sum_s(\eta_{j_s}-\bar{\eta}_{j_s}),\\
	0\leq\lambda_{j_l}\leq C_1,\; 0\leq\lambda_{h_i}\leq C_2, \eta_{j_s}\geq0,\;\bar{\eta}_{j_s}\geq0,
	\end{array}
	\]
	where $j\in\mathbb{N}_J,l\in\mathbb{N}_{l_j},h\in\mathbb{N}_H,i\in\mathbb{N}_{I_h},s\in\mathbb{N}_{s_j}$.

\end{proof}

\begin{definition}[Concave and Convex Functions]
	A function $f:X\subseteq\mathbb{R}^n\rightarrow\mathbb{R}$ is said to be
	\[
	\begin{array}{cc}
	\mbox{convex}\hspace{0.2cm}\mbox{iff} & f(\lambda x+(1-\lambda)y)\leq\lambda f(x)+(1-\lambda)f(y)\\
	\mbox{concave}\hspace{0.2cm}\mbox{iff}  & f(\lambda x+(1-\lambda)y)\geq\lambda f(x)+(1-\lambda)f(y)
	\end{array}
	\]
	for any $x,y\in X$, $\lambda\in[0,1]$.
\end{definition}

\begin{proof}[{\bf Proof of Lemma \ref{concpres}}]
	
	\begin{enumerate}
		\item This is obvious since the opposite of any convex function is a concave one and vice versa.
		\item If $f_\varphi$ and $f_\psi$ are concave, then for all $x,y,\,\lambda\in[0,1]$, $f_{\varphi\wedge\psi}(\lambda x+(1-\lambda)y)=\min\{f_\varphi(\lambda x+(1-\lambda)y),\,f_\psi(\lambda x+(1-\lambda)y)\}\geq\min\{\lambda f_\varphi(x)+(1-\lambda)f_\varphi(y),\,\lambda f_\psi(x)+(1-\lambda)f_\psi(y)\}\geq\lambda f_{\varphi\wedge\psi}(x)+(1-\lambda)f_{\varphi\wedge\psi}(y)$. 
		\\
		Moreover, by definition $f_{\varphi\oplus\psi}(x)=\min\{1,f_\varphi(x)+f_\psi(x)\}$, thus if $f_{\varphi\oplus\psi}(\lambda x+(1-\lambda)y)=1$ then obviously it is greater or equal than $\lambda f_{\varphi\oplus\psi}(x)+(1-\lambda)f_{\varphi\oplus\psi}(y)$. Otherwise $f_{\varphi\oplus\psi}=f_\varphi+f_\psi$ and sum preserves concavity (and it preserves convexity too) so the thesis easily follows.
		\item This point follows from 1) and 2) plus recalling that $f_{\varphi\vee\psi}=f_{\neg(\neg\varphi\wedge\neg\psi)}$ and $f_{\varphi\otimes\psi}=f_{\neg(\neg\varphi\oplus\neg\psi)}$.
	\end{enumerate}
\end{proof}

\begin{proof}[{\bf Proof of Proposition \ref{prop:conc}}]
	First of all we note that, as a consequence of Lemma \ref{concpres}, if $\varphi$ belongs to the concave fragment, then $f_\varphi$ is a concave function. Indeed, all the connectives occurring in $\varphi$ correspond to operations that preserve concavity and literals and constants correspond to affine functions. The same argument holds if the formula belongs to the convex fragment.
	
	On the other hand, let us suppose that $f_\varphi$ is a concave piecewise linear function, hence there exist some elements $a_{i_j},b_i\in\mathbb{Z}$ for $i=1,\ldots,m$ and $j=1,\ldots,n$, such that:		
	\[
	f_\varphi(x)=\min_{i=1}^m\,a_{i_1}x_1+\ldots+a_{i_n}x_n+b_i,\qquad x\in[0,1]^n.
	\]
	If we set $p_i(x)=a_{i_1}x_1+\ldots+a_{i_n}x_n+b_i$ for $i=1,\ldots,m$, our claim follows provided every $p_i$ corresponds to a formula in $(\wedge,\oplus)^*$. Indeed the operation of minimum is exactly performed by the connective $\wedge$. Let us fix $i\in\{1,\ldots,m\}$, then we can write
	\[
	p_i(x)=\sum_{j\in P_i}a_{i_j}x_j+\sum_{j\in N_i}a_{i_j}x_j+b_i,
	\]
	where $P_i=\{j\leq n:\,a_{i_j}>0\}$ and $N_i=\{j\leq n:\,a_{i_j}<0\}$. For short, given any $\psi\in{\bf \L}$, we write $a\psi$ with the meaning of $\bigoplus_{i=1}^a \psi$ or $\underbar{0}$, if $a>0$ or $a=0$ respectively. Therefore, we can consider the following formula as corresponding to the function $p_i$:
	\[
	\varphi_i=\bigoplus_{j\in P_i}a_{i_j}x_j\oplus\bigoplus_{j\in N_i}|a_{i_j}|\neg x_j\oplus q_i\underbar{1}.
	\]
	Indeed the first strong disjunction corresponds to all the positive monomials of $p_i$. The second one corresponds to all the negative monomials of $p_i$, but it also introduces the quantity $\sum_{j\in N_i}|a_{i_j}|$. Finally $q_i=b_i-\sum_{j\in N_i}|a_{i_j}|$, with $q_i\geq0$ since $p_i(x)\geq0$ for all $x\in[0,1]^n$ and in particular $p_i(\bar{x})=b_i-\sum_{j\in N_i}|a_{i_j}|\geq0$ where $\bar{x}$ is the vector with 0 in positive and 1 in negative monomial positions respectively. The overall formula can be written as $\varphi=\varphi_1\wedge\ldots\wedge\varphi_m$.
\end{proof}

%
%

\ifCLASSOPTIONcaptionsoff
  \newpage
\fi



%

\bibliography{references}
\bibliographystyle{IEEEtran}

\end{document}